\documentclass{article}

\expandafter\let\csname equation*\endcsname\relax
\expandafter\let\csname endequation*\endcsname\relax
\usepackage{amsmath}
\usepackage{amsthm}
\usepackage{amsfonts}
\usepackage{amssymb}
\usepackage{graphicx}
\usepackage[margin = 1.4in]{geometry}

\newtheorem{Thm}{Theorem}
\newtheorem{Prop}[Thm]{Proposition}
\newtheorem{Cor}[Thm]{Corollary}
\newtheorem{Lem}[Thm]{Lemma}

\newtheorem{Def}[Thm]{Definition}

\theoremstyle{definition}

\newtheorem{Ex}[Thm]{Example}

\newcommand{\BPP}{bounded parameter property}

\newcommand{\bpp}{b.p.p}

\newcommand{\man}[1]{\ensuremath{\mathcal{#1}}}
\newcommand{\cur}[1]{\ensuremath{\mathcal{#1}}}
\newcommand{\AB}[1][M]{\ensuremath{\mathcal{B}(\man{#1})}}
\newcommand{\covers}{\ensuremath{\rhd}}
\newcommand{\Bound}[1]{\ensuremath{\textnormal{B}({#1})}}
\newcommand{\EmptySet}{\ensuremath{\varnothing}}
\newcommand{\cleq}[1]{\ensuremath{{#1}^\leq}}
\newcommand{\cfil}[1]{\ensuremath{\cur{C}\left({#1}\right)}}

\newcommand{\App}[2]{\ensuremath{\textnormal{App}({#1},{#2})}}
\newcommand{\Unapp}[2]{\ensuremath{\textnormal{Nonapp}({#1},{#2})}}
\newcommand{\AppBound}[2]{\ensuremath{\textnormal{App}_{\textnormal{\scriptsize  Sing}}({#1},{#2})}}
\newcommand{\AppUnbound}[2]{\ensuremath{\textnormal{App}_{\textnormal{\scriptsize  Inf}}({#1},{#2})}}
\newcommand{\Env}[1][\man{M}]{\ensuremath{\Phi({#1})}}
\newcommand{\InexCom}[1][\man{M}]{\ensuremath{\text{NPre}({#1})}}
\newcommand{\BPPSet}[1][\man{M}]{\ensuremath{\text{BPP}({#1})}}
\newcommand{\BPPInexCom}[1][\man{M}]{\ensuremath{\text{NPre}_\textnormal{\scriptsize\hspace{1pt}\bpp.}({#1})}}

\begin{document}

\title{The dependence of the abstract boundary classification on a set of curves I: 
An algebra of sets on bounded parameter property satisfying sets of curves}

\author{B E Whale\footnote{bwhale@maths.otago.ac.nz} \footnote{Department of Mathematics and Statistics, University of Otago}}

\maketitle

\begin{abstract}
	The abstract boundary uses sets of curves with the bounded parameter
  property (b.p.p.) to classify the elements of the abstract boundary
  into regular points, singular points, points at infinity and so on.
  To study how the classification changes as this set of curves is changed
  it is necessary to describe the relationships between these sets of curves in 
  a way that reflects the effect of the curves on the classification. The
  usual algebra of sets fails to do this. We remedy this situation by
  generalising inclusion, intersection, and union: producing an algebra
  of sets on the set of all b.p.p. satisfying sets of curves that does
  appropriately describe the relative effects on the classification.
  In Part II we use this algebra of sets to show how the classification
  changes as the set of b.p.p. satisfying set of curves is changed with respect
  to this generalization.
\end{abstract}

\section{Introduction}\label{Intro}\label{sec.Background}
  A boundary construction in General Relativity
  is a method to attach
  `ideal' points to a Lorentzian manifold. The constructions are 
  designed so that the ideal points can be classified 
  into physically
  motivated classes such as regular points, singular points, points
  at infinity and so on.

  To do this most boundary constructions use, implicitly or explicitly, 
  a set of curves, usually with a particular type of parametrization. 
  For example the $g$-boundary, \cite{Geroch1968a}, relies on incomplete 
  geodesics with affine parameter, 
  the $b$-boundary, \cite{Schmidt1971}, on incomplete curves with generalized 
  affine parameter and the $c$-boundary, \cite{GerochPenroseKronheimer1972},
  and its modern variants,
  \cite{Marolf2003New,citeulike:8211160,Flores2010Final}, on
  endless causal curves. 
  
  Given the importance, however, of providing physical 
  interpretations of boundary points and the necessity of using curves to do 
  so, even if no choice of curves is forced by a boundary construction 
  a choice will need to be made
  at some point for physical applications. 
  The boundary construction considered in this paper, the abstract boundary
  \cite{ScottSzekeres1994}, is an example of this.
  It does not use a set of curves for its construction but it does use
  a set of curves for its classification.
  It allows any set of curves to be used
  as long as the set
  satisfies the bounded parameter property (\bpp.), 
  see Definition \ref{def.BPP}. The \bpp.\ ensures
  that a consistent physical interpretation can be applied to
  abstract boundary points.
  
  Papers such as \cite{ScottSzekeres1994,Geroch1868} 
  reiterate the point that careful consideration of the set of curves
  used in a classification of boundary points is 
  needed to get a correct definition
  of a singularity. Indeed, the issues with giving a consistent physical 
  interpretation, raised by the non-Hausdorff and non-$T_1$ separation
  properties of the $g$-, $b$- and older $c$-boundaries,
  is related to the set of curves used for classification `being too big', 
  e.g. including precompact timelike geodesics. For these boundaries, as the
  set of curves is also 
  connected to the construction of the boundary points, the
  inclusion of `too many curves' is part of the root cause of these
  separation properties, 
  \cite{Ashley2002a,Whale2010}. For example the non-Hausdorff behaviour
  of the $b$-boundary is directly related
  to the existence of inextendible incomplete curves that have more than
  one limit point, 
  \cite[Proposition 8.5.1]{HawkingEllis1973}.
  
  Thus the study of how sets of curves 
  associated to boundary constructions affects the construction and
  classification of boundary points is a fundamental, if over looked,
  component of the study of boundary constructions, and
  therefore singularities, in General Relativity.
  This begs the question of how the classification of abstract boundary points
  changes when the b.p.p.\ satisfying set of curves is changed. We present
  the answer to this question in the two parts of this series of papers.

  Given a b.p.p.\ satisfying set of curves, $\cur{C}$, the boundary, 
  $\partial\phi(\man{M})=\overline{\phi(\cur M)}-\phi(\cur M)$, of an embedding
  of a manifold $\phi:\man{M}\to\man{M}_\phi$, $\dim \man{M}=\dim\man{M}_\phi$,
  can be divided into three sets; $\AppBound{\phi}{\cur{C}}$, the set of
  points which are approach by at least one curve with bounded parameter;
  $\AppUnbound{\phi}{\cur{C}}$, the set of points which are only approached
  by curves with unbounded parameters; and $\Unapp{\phi}{\cur{C}}$, the set of points that
  are not approached. This division is used to produce the classification of
  the abstract boundary points represented by the elements of $\partial\phi(\man{M})$,
  see Section 1.1 and Proposition 16 of \cite{Whale2012b}.
  
  In order to study how the classification changes when the b.p.p.\ satisfying
  set of curves is changed it is necessary to have a way of relating
  b.p.p.\ satisfying sets that gives information about the relationship
  of the induced division of $\partial\phi(\man{M})$, for all $\phi\in\Env$.
  The usual algebra of sets (inclusion, intersection and union) fails
  to give all information of this type. In particular, there can exist curves
  in a b.p.p.\ satisfying set that do not contribute to the division of
  $\partial\phi(\man{M})$, for all $\phi\in\Env$. Such curves play no role
  in the classification of boundary or abstract boundary points, yet can
  prevent the intersection or union of b.p.p.\ satisfying sets from being
  well defined.
  Hence, before we can analyse how the classification
  changes,
  we need a more appropriate way of comparing b.p.p.\ satisfying sets of curves.
  
  This problem is solved in  three steps.
  First, we show that there is a 
  correspondence between the set of all \bpp.\ satisfying sets of
  curves and the set of all sets of non-precompact curves
  that have `compatible' parameters, see Definition \ref{def.overline}.
  Second, we explore this correspondence and demonstrate that, with respect to
  suitable equivalence relations, the induced correspondence is 
  bijective. These equivalence relations ensure that sets of curves
  are identified only if they produce the same classification of
  abstract boundary points.
  That is, we prove that non-precompact sets of
  curves, with compatible parameters,
  can be used as a replacement for \bpp.\ satisfying sets of curves
  in the abstract boundary classification. 
  Third, we use this result to give a generalization of the 
  algebra of sets restricted to the set of all b.p.p.\ satisfying sets
  of curves. This generalization solves the problem discussed in the 
  previous paragraph.
  
  The paper is divided into four sections. This section continues
  with preliminary definitions. Section \ref{sec.main} introduces
  the relationship between \bpp.\ satisfying sets 
  of curves and sets of non-precompact curves with compatible
  parameters.
  Section \ref{Sc one-to-one correspondences} presents two pairs of 
  equivalence relations under which the correspondence of Section \ref{sec.main}
  becomes a bijection. 
  The first pair is easy to use in applications, but there exist
  pairs of b.p.p.\ satisfying sets of curves which give the same
  division of $\partial\phi(M)$, for all $\phi\in\Env$,
  which are not identified. 
  The second pair gives the 
  appropriate equivalence relations for describing
  when two sets of b.p.p.\ satisfying sets of curves produce 
  the same division of $\partial\phi(M)$, for all $\phi\in\Env$,
  but they will be difficult to
  use in practice.
  In Section \ref{sec:realtions} we complete this work by;
  giving the claimed generalizations of inclusion, intersection and union;
  proving that the generalizations form an algebra of sets; and
  proving that the division of $\partial\phi({M})$ induced by
  b.p.p.\ satisfying sets related via our generalization of the algebra of
  sets are similarly related by the usual algebra of sets.

\subsection{Preliminary definitions}

  We shall only consider manifolds, \man{M}, that are
  paracompact, Hausdorff, connected, $C^\infty$-mani-folds.
  A precompact set is one whose closure is compact.

\begin{Def}[{\cite[Definition 9]{ScottSzekeres1994}}]
  An embedding, $\phi:\man{M}\to\man{M}_\phi$, of $\man{M}$ is an envelopment 
  if $\man{M}_\phi$ has 
  the same dimension as $\man{M}$.  Let $\Env$ be the set of all envelopments 
  of $\man{M}$.
\end{Def}

\begin{Def}[{\cite[Definition 14 and 22, Theorem 18]{ScottSzekeres1994}}]
  Let $\Bound{\man{M}}$ be the set of all ordered pairs $(\phi,U)$
  of envelopments $\phi$ and subsets $U$ of 
  $\partial\phi(\man{M})=\overline{\phi(\man{M})}-\phi(\man{M})$. That is,
  \[
    \Bound{\man{M}}=\{(\phi,U): \phi\in\Env,\, U\subset\partial\phi(\man{M})\}.
  \]
  Define an partial order $\covers$ on 
  $\Bound{\man{M}}$ by $(\phi, U)\covers(\psi, V)$ if and only 
  if for every sequence $\{x_i\}$ in \man{M}, $\{\psi(x_i)\}$ has a 
  limit point in $V$ implies that $\{\phi(x_i)\}$ has a limit point in $U$.  
  We can construct an equivalence relation $\equiv$ on $\Bound{\man{M}}$ by 
  $(\phi, U)\equiv(\psi, V)$ if and only if $(\phi, U)\covers(\psi, V)$ and
  $(\psi, V)\covers(\phi, U)$.    Denote the equivalence class of $(\phi,U)$ by 
  $[(\phi, U)]$.
  
  The abstract boundary is the set
  \[
    \AB=\left\{[(\phi,U)]\in\frac{\Bound{\man{M}}}{\equiv}:\exists 
    (\psi,\{p\})\in[(\phi,U)] \right\}.
  \]
  It is the set of all equivalence classes of $\Bound{\man{M}}$ under the 
  equivalence relation $\equiv$ that contain an element $(\psi,\{p\})$ where 
  $p\in\partial\psi(\man{M})$.   The elements of the abstract boundary are referred to as abstract boundary
  points.
\end{Def}

Abstract boundary points behave in many ways like `normal' points,
\cite{FamaClarke1998,FamaScott1994}.
Following the lead of \cite{ScottSzekeres1994} we work with $C^0$ 
piecewise $C^1$ curves. 

\begin{Def}[{\cite[Definition 1, 2, and 3]{ScottSzekeres1994}}]
\label{def.curves}
A parametrized $C^0$ piecewise $C^1$ curve (or curve) $\gamma$ in the manifold 
\man{M} is a $C^0$ map
$\gamma:[a,b)\to\man{M}$ where $[a,b)\subset\mathbb{R}\cup\{\infty\}$, 
$a<b\leq\infty$ with a finite subset $a=\tau_0,\tau_1,\ldots,\tau_m=b$ so that 
on each segment $[\tau_i,\tau_{i+1})$
the tangent vector $\gamma':[\tau_i,\tau_{i+1})\to T\man{M}$ is 
everywhere non-zero.			
We shall say that $\gamma$ is bounded if $b<\infty$ otherwise $\gamma$ is 
unbounded.

A curve $\delta:[a',b')\to\man{M}$ is a subcurve of $\gamma$ if 
$a\leq a'< b' \leq b$ and $\delta=\gamma|_{[a',b')}$.  
That is a curve $\delta$ is a subcurve of $\gamma$ if $\delta$ is the 
restriction of $\gamma$ to some right-half open interval of
$[a,b)$. We shall denote this by $\delta\leq\gamma$. If $a'=a$ and $b'<b$ we 
shall say that $\gamma$ is an extension of $\delta$. 
A curve $\delta$ is
inextendible if there does not exist an extension of $\delta$.

A change of parameter is a monotone increasing surjective $C^1$ function, 
$s:[a',b')\to[a,b)$.  The curve $\delta$ is obtained 
from the curve $\gamma$ if $\delta = \gamma \circ s$. 

The curve
$\gamma:[a,b)\to \man{M}$ is precompact if $\overline{\gamma([a,b))}$ is
compact in $\man{M}$. That is $\gamma$ is precompact if its image is precompact.
\end{Def}
  
The condition $\delta=\gamma|_{[a',b')}$ requires the two curves to be
equivalent as functions, on the appropriate domain. This implies that 
the definition of $\leq$ considers the parametrization chosen for $\gamma$.  
That is, if $\delta(t)\leq\gamma(t)$
then we know that
$\delta(t)\not\leq\gamma(2t)$. Thus
we draw a distinction between different parametrizations of the same image of a curve. 
    
\begin{Def}[{\cite[Definition 4]{ScottSzekeres1994}}]\label{def.BPP}
  A set $\cur{C}$ of parametrized curves is said to have 
  the \BPP\ (or \bpp.) 
  if the following conditions are satisfied;
    
  \begin{enumerate}
    \item For all $p\in\man{M}$ there exits $\gamma:[a,b)\to\man{M}\in\cur{C}$ 
    so that $p\in\gamma([a,b))$.\label{def.BPP.1}
    \item If $\gamma\in\cur{C}$ and $\delta\leq\gamma$ then $\delta\in\cur{C}$.\label{def.BPP.2}
    \item For all $\gamma,\delta\in\cur{C}$, if $\delta$ is obtained from 
    $\gamma$ by a change of parameter
      then either both curves are bounded or both are unbounded.\label{4.3}\label{def.BPP.3}
  \end{enumerate}
\end{Def}
  
The classification of abstract boundary points, \cite[Section 4 and 5]{ScottSzekeres1994},
is given by a classification of boundary points -- that is
elements of $\partial\phi(\man{M})$ for each
$\phi\in\Env$ -- and then by showing which `parts' of this classification
are invariant under the equivalence relation $\equiv$. 
The classification of boundary points of an envelopment, $\phi\in\Env$,
depends on an analysis of the relative bounded/un-boundedness
of curves in a \bpp.\ satisfying
set that have common limit points in $\partial\phi(\man{M})$ as well as  
certain properties of
$\covers$ and the Lorentzian metric. 
Here we are only concerned with the structure of
\bpp.\ satisfying sets necessary for the classification and
therefore only concerned with the relative bounded/un-boundedness
of elements of our \bpp.\ satisfying set, \cite[Section 4]{ScottSzekeres1994}.
  
\begin{Def}[{\cite[Definition 23]{ScottSzekeres1994}}]\label{approachable-points}
  Let $\phi\in\Env$ and $\cur{C}$ be a set of curves with the \bpp. A 
  boundary point $p\in\partial\phi(\man{M})$ is approachable if 
  there exists $\gamma\in\cur{C}$ so that $p$ is a limit point 
  of the curve $\phi\circ\gamma$. 
  We make the following definitions;
  \begin{align*}
    \App{\phi}{\cur{C}}&=\{p\in\partial\phi(\man{M}):p
      \text{ is approachable}\}\\
    \Unapp{\phi}{\cur{C}}&=\{p\in\partial\phi(\man{M}):p
      \text{ is not approachable}\}\\
      &=\partial\phi(\man{M})-\App{\phi}{\cur{C}}.
  \end{align*}
\end{Def}

\begin{Def}\label{singinfpoints}
  Let $\AppBound{\phi}{\cur{C}}$ and $\AppUnbound{\phi}{\cur{C}}$ 
  be defined by, 
  \begin{multline*}
    \AppBound{\phi}{\cur{C}}=\{p\in\App{\phi}{\cur{C}}:\text{ there
    exists }\\ \text{ a bounded }
    \gamma\in\cur{C}\text{ with } 
    p\text{ as a limit point of}\ \phi\circ\gamma\}
  \end{multline*}
  \vspace{-1.005cm}
  \begin{multline*}
    \AppUnbound{\phi}{\cur{C}}=\{p\in\App{\phi}{\cur{C}}:
    \text{ for all }
    \gamma\in\cur{C}\\ \text{ so that }
    p\text{ is a limit point of}\ \phi\circ\gamma,\, \gamma\text{ is unbounded}\}.
  \end{multline*}
\end{Def}

  We use the subscripts `Sing' and `Inf' as the corresponding sets are 
  closely related to the abstract boundary definitions of singularities
  and points at infinity respectively, 
  \cite[Definitions 31 and 37]{ScottSzekeres1994}. 
  If $\man M$ is the 
  Schwarzschild spacetime and $\cur{C}$ is the set
  of all affinely parametrized null geodesics then,
  with respect to the envelopment given by the Penrose-Carter
  conformal compactification, 
  a point of the
  singularity is an element of \AppBound{\phi}{\cur{C}}
  and a point of future timelike infinity
  is an element of \AppUnbound{\phi}{\cur{C}}.
  
  If two \bpp.\ satisfying sets $\cur{C}$ and $\cur{D}$ are such that,
  for all $\phi\in\Env$, 
  $\AppBound{\phi}{\cur{C}}=\AppBound{\phi}{\cur{D}}$ and 
  $\AppUnbound{\phi}{\cur{C}}=\AppUnbound{\phi}{\cur{D}}$ then the 
  classifications of abstract boundary points given by $\cur{C}$ and $\cur{D}$
  will be the same, \cite[Section 4 and 5]{ScottSzekeres1994}.

\section{Non-precompact curves and b.p.p.\ satisfying sets of\newline curves}
\label{sec.main}

In this section we show that every \bpp.\ satisfying set corresponds to a set 
of compatible non-precompact curves and vice versa.
The compatibility condition is related to Condition \eqref{4.3} of Definition 
\ref{def.BPP}.

\begin{Def}\label{BPPSet}
	Let $\BPPSet$ be the set of all sets of curves with the \bpp. That is
	$\BPPSet=\{\cur{C}: \cur{C}$ is a set of curves with the \bpp.$\}.$
\end{Def}

\begin{Def}\label{def.overline}
  Let $\InexCom$ be the set of non-precompact 
  curves. Let 
  $\BPPInexCom$ be the set of subsets, $S$, of 
  $\InexCom$ such that for all 
  $\gamma:[a,b)\to\man{M},\,\delta:[p,q)\to\man{M}\in S$, if there exists 
  $c\in[a,b),r\in[p,q)$ and a change of parameter $s:[r,q)\to[c,b)$ so that 
  $\gamma|_{[c,b)}\circ s=\delta|_{[r,q)}$ then either both 
  $\gamma$ and $\delta$ are bounded or both
  are unbounded. We say that the elements of $S$ have compatible parameters.
  We allow $\EmptySet\in\BPPInexCom$.
\end{Def}

In effect Definition \ref{def.overline} says that two curves that eventually
have the same image must either both be bounded or both unbounded.
With this notation we can now give the first of two functions which give the
correspondence, mentioned in the introduction.

\begin{Prop}\label{prop:fdef}
  There exists a function $f:BPP(\man{M})\to \BPPInexCom$, given 
  by 
  $f(\cur{C})=\InexCom\cap\cur{C}$.
\end{Prop}
\begin{proof}
  The function is well defined and by definition of $\cur{C}$ we 
  know that $\InexCom\cap\cur{C}\in\BPPInexCom$.
\end{proof}

It takes slightly more work to define our other function.

\begin{Def}
  Given a set of curves $S$, let 
  $P(S)=\{p\in\man{M}:\exists\gamma:[a,b)\to\man{M}\in S$ so that 
  $p\in\gamma([a,b))\}$.
\end{Def}

Thus $P(S)$ is the set of all points in the manifold that are contained in 
the image of some curve in $S$.

For each $p\in \man M$ choose $V_p$ a precompact 
open normal neighbourhood of $p$ and choose $v_p\in T_p\man{M}$. Let 
$\gamma_p:[-\epsilon_p,\epsilon_p)\to\man{M}$, for some
$\epsilon_p>0$ depending on $p$, be a geodesic lying in $V_p$ so that 
$\gamma_p(0)=p$ and $\gamma_p'(0)=v_p$.

\begin{Def}\label{Cnormdef}  
  Let
  $\cfil{U}=\{\delta$ a curve 
  $: \exists p\in U$ so that $\delta\leq\gamma_p\}$.
\end{Def}

\begin{Def}\label{curves less than}
  Let $S$ be a set of curves. Define 
  $
    \cleq{S} = \{\delta\text{ a curve}:\exists\gamma\in S,\,\delta\leq\gamma\}.
  $
\end{Def}

\begin{Prop}\label{gfunc}
	There exists a function $g:\BPPInexCom\to BPP(\man{M})$, given
	by 
  $
    g(S)=\cleq{S}\cup
    \cfil{\man{M}-P(S)}.
  $
\end{Prop}
\begin{proof}
  As before, it is clear that $g$ is well defined, so we must only check that 
  $g(S)$ satisfies the \bpp. We must check the three 
  conditions of Definition \ref{def.BPP}. Conditions $1$ and $2$
  follow directly from Definitions \ref{Cnormdef} and \ref{curves less than}.
  
  We now show that Condition $3$ holds on $g(S)$. Let 
  $\alpha:[a,b)\to\man{M},\beta:[p,q)\to\man{M}\in g(S)$ 
  be such that $\alpha$ and $\beta$ are obtained from each other by 
  a change of parameter $s:[p,q)\to[a,b)$, so that $\alpha\circ s = \beta$.
  If 
  $\alpha,\beta\in \cfil{\man{M}-P(S)}$ then 
  by definition \ref{Cnormdef} we know that both must be bounded. 
  
  Suppose that $\alpha\in\cfil{\man{M}-P(S)}$ and
  $\beta\in\cleq{S}$. 
  Since 
  $\beta\in \cleq{S}$ 
  there exists
  $\gamma_\beta:[p',q')\to\man{M}\in S$ so that $\beta\leq\gamma_\beta$, that is 
  $[p,q)\subset[p',q')$ and $\beta=\left.\gamma_\beta\right|_{[p,q)}$. 
  Since $\gamma_\beta\in S$ it must be the case that 
  $\overline{\gamma_\beta([p',q'))}$ is not 
  compact. 
  As $\alpha\in\cfil{\man{M}-P(S)}$ we know 
  that $\overline{\alpha([a,b))}$ is necessarily compact.
  Hence as 
  $\overline{\alpha([a,b))}=\overline{\alpha\circ s([p,q))}=
  \overline{\beta([p,q))}$ we know that 
  $\overline{\beta([p,q))}$ is compact. 
  Since $\overline{\beta([p,q))}\subset\overline{\gamma_\beta([p',q'))}$ 
  and $\overline{\gamma_\beta([p',q'))}$ is not compact we can conclude that 
  $[p,q)$ is a proper subset of $[p',q')$ and in particular
  that $q<q'$. Hence $q\in\mathbb{R}$ and therefore $\beta$ must be bounded.
  
  Now suppose 
  $\alpha,\beta\in
  \cleq{S}$.  
  As $\alpha\in\cleq{S}$ there 
  exists
  $\gamma_\alpha:[a',b')\to\man{M}\in S$ so that $\overline{\gamma_\alpha([a',b'))}$ 
  is not compact, $[a,b)\subset[a',b')$ 
  and $\alpha=\left.\gamma_\alpha\right|_{[a,b)}$. Let 
  $\gamma_\beta:[p',q')\to\man{M}$ be as above.
  
  If $q<q'$ then $\overline{\beta([p,q))}$ is compact and
  $\overline{\beta([p,q))}=\overline{\alpha\circ s([p,q))}=
  \overline{\alpha([a,b))}$ is also compact. Therefore
  $b<b'$ and $\alpha$ is bounded. Applying the same argument for $\alpha$ 
  shows that
  $q<q'$ if and only if $b<b'$. Thus, in this case, both $\beta$ and 
  $\alpha$ are bounded.
  
  Hence, to complete the proof, we need only consider the case when $q=q'$ and $b=b'$.
  We see that $\gamma_\beta|_{[p,q')}=\beta=\alpha\circ s = \gamma_\alpha|_{[a,b')}
  \circ s$. Since $\gamma_\alpha,\gamma_\beta\in S$
  from Definition \ref{def.overline} we know 
  that $\gamma_\alpha$ and $\gamma_\beta$ are either both 
  bounded or
  both unbounded.  Since $q=q'$ and $b=b'$ the same applies to $\alpha$ and 
  $\beta$.  That is either both are
  bounded or both are unbounded, as required.\qedhere
\end{proof}

From Definitions \ref{approachable-points} and \ref{singinfpoints} we see that
a curve 
in a \bpp.\ satisfying set provides information for the abstract boundary
classification if and only if the curve has a limit point in 
$\partial\phi(\man M)$
for some $\phi\in\Env$, see \cite[Section 4]{ScottSzekeres1994}.
That is if and only if the curve is not pre-compact. 
Hence we have the following.
\begin{Prop}\label{prop.goodtranslation}
  Let $\cur{C}\in\BPPSet$ then
  \begin{gather*}
    \AppBound{\phi}{\cur{C}}=\AppBound{\phi}{g\circ f(\cur{C})}\\
    \AppUnbound{\phi}{\cur{C}}=\AppUnbound{\phi}{g\circ f(\cur{C})}
  \end{gather*}
\end{Prop}
\begin{proof}
  Let $p\in\App{\phi}{\cur{C}}$. Then there exists $\gamma\in\cur{C}$
  so that $p$ is an accumulation point of $\phi\circ\gamma$.
  This implies that $\gamma$ is not pre-compact in $M$. Hence
  $\gamma\in f(\cur{C})$. Since $\gamma\in f(\cur{C})$
  we know that $\gamma\in (f(\cur{C}))^\leq$ so that
  $\gamma\in g\circ f(\cur{C})$. This implies that
  $p\in\App{\phi}{g\circ f(\cur{C})}$. If $p\in\AppBound{\phi}{\cur{C}}$
  then $\gamma$ can be taken to be bounded so that
  $p\in\AppBound{\phi}{g\circ f(\cur{C})}$.
  
  Likewise if $p\in\App{\phi}{g\circ f(\cur{C})}$, there exists
  $\gamma\in g\circ f(\cur{C})$ so that $p$ is an accumulation point
  of $\phi\circ\gamma$. Since $p\in\partial\phi(M)$ this implies that
  $\gamma$ is not pre-compact and therefore $\gamma\in (f(\cur{C}))^\leq$.
  Thus there exists $\delta\in f(\cur{C})$ so that $\gamma\leq\delta$.
  This implies that $\delta\in \cur{C}$ and hence
  $p\in\App{\phi}{\cur{C}}$. If $p\in\AppBound{\phi}{g\circ f(\cur{C})}$
  then $\gamma$ can be taken to be bounded. Hence $\delta$ is bounded
  and $p\in \AppBound{\phi}{\cur{C}}$.
  
  Thus we have that $\App{\phi}{\cur{C}}=\App{\phi}{g\circ f(\cur{C})}$
  and $\AppBound{\phi}{\cur{C}}=\AppBound{\phi}{g\circ f(\cur{C})}$
  which implies that $\AppUnbound{\phi}{\cur{C}}=
  \AppUnbound{\phi}{g\circ f(\cur{C})}$,
  as required.
\end{proof}
 
Hence when analysing the properties of the abstract boundary classification
it is sufficient to use $f$ and $g$ to `translate' between 
\bpp.\ satisfying sets and elements of $\BPPInexCom$. 
We give three results describing the relationship 
between $f$ and $g$.

\begin{Lem}\label{fdinverse}
  Let $S\in{\BPPInexCom}$ then 
  $S\subset f\circ g(S)$ and if $\gamma:[a,b)\to\man{M}\in f\circ g(S)-S$ 
  then there exists
  $\delta:[p,q)\to\man{M}\in S$ so that $\gamma\leq\delta$ and 
  $\gamma\neq\delta$. In particular,
  $q=b$ and $p\neq a$.
\end{Lem}
\begin{proof}
  We first note that
  \begin{align*}
    f\circ g(S)&= g(S)\cap\InexCom\\
      &=\Bigl(\cleq{S}\cup
        \cfil{\man{M}-P(S)}\Bigr)\cap\InexCom\\
      &=\Bigl(\cleq{S}\cap\InexCom\Bigr)
        \cup \Bigl(\cfil{\man{M}-P(S)}
        \cap\InexCom\Bigr).
  \end{align*} 
    Since $S\subset 
    \cleq{S}\cap\InexCom$ 
    we know that
    $S\subset f\circ g(S)$.    
    
    Let
    $\gamma:[a,b)\to\man{M}\in f\circ g(S)-S$. 
    By Definitions \ref{def.overline} and \ref{Cnormdef}
    the intersection $\cfil{\man{M}-P(S)}\cap
    \InexCom$ is empty, therefore
    $\gamma\in
    \cleq{S}\cap\InexCom$.
    Hence there exists $\delta:[p,q)\to\man{M}\in S$ so that $\gamma\leq\delta$. 
    By definition we know that $b\leq q$.
    If $b<q$ then $\overline{\gamma([a,b))}$ must be compact and hence 
    $\gamma\not\in\InexCom$. Since
    this is a contradiction it must be the case that $b=q$. If $a=p$ we would 
    have $\gamma=\delta$ and thus
    $\gamma\in S$. Since this is also a contradiction it must be the case 
    that $p<a$, as required.
\end{proof}

\begin{Prop}\label{lem.fdinvers}
  Let $\cur{C}\in\BPPSet$ then $f\circ g\circ f(\cur C)=f(\cur C)$.
\end{Prop}
\begin{proof}
  We can calculate that,
  \begin{align*}
    f\circ g\circ f(\cur{C}) &= f\Bigl(\cleq{f(\cur{C})}\cup
      \cfil{\man{M}-P(f(\cur{C}))}\Bigr)\\
    &= \Bigl(\cleq{f(\cur{C})}\cup
      \cfil{\man{M}-P(f(\cur{C}))}\Bigr)
      \cap\InexCom\\
    &= \cleq{f(\cur{C})}\cap\InexCom.
  \end{align*}
  Let $\gamma\in f\circ g\circ f(\cur{C})$.
  Lemma \ref{fdinverse} implies that
  there exists $\delta\in f(\cur{C})$ so that $\gamma\leq\delta$.
  Since $\delta\in f(\cur{C})=\cur{C}\cap\InexCom$ we know that 
  $\delta\in\cur{C}$. By definition of the \bpp.\ we then know that
  $\gamma\in\cur{C}$. Since
  $\gamma\in f\circ g\circ f(\cur{C})\subset \InexCom$, Proposition \ref{prop:fdef}
  implies that $f\circ g\circ f(\cur{C})\subset f(\cur{C})$. From 
  Lemma \ref{fdinverse}
  we know that $f(\cur{C})\subset f\circ g\circ f(\cur C)$ as required.
\end{proof}
  
\begin{Prop}\label{prop:rel3}
  Let $S\in\BPPInexCom$ then $g\circ f\circ g(S)=g(S)$.
\end{Prop}
\begin{proof}
  We first show that $P\bigl(f\circ g(S)\bigr)=P(S)$. Let $x\in 
  P(f\circ g(S))$
  then there exists $\gamma:[a,b)\to\man{M}\in f\circ g(S)$ so that
  $x\in\gamma([a,b))$. By Lemma \ref{fdinverse} we know that there
  exists $\delta:[p,b)\to\man{M}\in S$ so that $\gamma\leq\delta$.
  Hence $x\in\delta([p,q))$ and 
  $x\in P(S)$. Thus $P\bigl(f\circ g(S)\bigr)\subset P(S)$. 
  Again, from Lemma \ref{fdinverse} we know that
  $S\subset f\circ g(S)$ and therefore $P(S)\subset P\bigl(f\circ g(S)\bigr)$
  as required.
  
  We now show that 
  $
    \cleq{S}=
    \cleq{f\circ g(S)}.
  $
  From Lemma \ref{fdinverse} we know that $S\subset f\circ g(S)$ so that
  $
    \cleq{S}\subset
    \cleq{f\circ g(S)}.
  $
  Let $\delta\in \cleq{f\circ g(S)}$ then there
  exists $\gamma\in f\circ g(S)$ so that $\delta\leq\gamma$. Since $f\circ g(S)=g(S)\cap\InexCom$
  we know that $\gamma\in \cleq{S}\cup
  \cfil{\man{M}-P(S)}$
  and that $\gamma\in\InexCom$. By Definition \ref{Cnormdef} we know that
  $\InexCom\cap \cfil{\man{M}-P(S)}=\EmptySet$.
  Hence $\gamma\in\cleq{S}$. Thus $\delta\in\cleq{S}$ as required.
  
  These two results give us 
  $
    g\circ f\circ g(S)=\cleq{f\circ g(S)}
    \cup
    \cur{C}_\textnormal{\scriptsize Norm}
    \left(\man{M}-P\bigl(f\circ g(S)\bigr)\right)=
    \cleq{S}\cup
    \cfil{\man{M}-P(S)}=g(S),
  $
  as required.
\end{proof}

These results imply that  $g$ is injective when restricted to the range 
of $f$ and that $f$ is injective when restricted to the range of $g$. 
Neither function is surjective, however. 

\subsection{The problem with the standard algebra of sets on b.p.p.\ satisfying sets}\label{subsec.examples}

We are now in a position to describe why the standard algebra of sets, restricted
to $\BPPSet$, is not good enough to investigate changes of the 
classification of abstract boundary points when the b.p.p.\ satisfying set
of curves changes.

From above we know that a b.p.p.\ satisfying set $\cur{C}$ affects the classification of
abstract boundary points though the induced
division of the boundary $\partial\phi(\man{M})$,
of an envelopment $\phi\in\Env$, into the sets $\AppBound{\phi}{\cur{C}}$, 
$\AppUnbound{\phi}{\cur{C}}$ and $\Unapp{\phi}{\cur{C}}$.
Only the relative bounded / unboundedness of curves in $f(\cur{C})$ 
contribute to this division.
Condition \eqref{def.BPP.2} of Definition \ref{def.BPP} uses $\leq$. As mentioned after 
Definition \ref{def.curves} this relation distinguishes between $\gamma(t)$
and $\gamma(2t)$ where $\gamma$ is a curve in $M$, however no distinction
between $\gamma$ and $\gamma(2t)$ need be made for the division of $\partial(\man{M})$.
Condition (1) of 
Definition \ref{def.BPP} can be satisfied using precompact curves (as
can be seen from the definition of $g$). Proposition \ref{prop.goodtranslation}
shows that precompact curves don't contribute to the division of $\partial(\man{M})$.
The problems with the standard algebra of sets on $\BPPSet$ stem
from these two conditions. To illustrate this we provide three examples.

First, since the definition of the \bpp.\ refers to $\leq$, but only considers
the relative bounded / unboundedness of curves, we have the following
problem. Let $\cur{C}\in\BPPSet$ and let 
$\hat{\cur{C}}=\{\gamma(2t):\gamma(t)\in\cur{C}\}$.
Then it is the case that $\hat{\cur{C}}\in\BPPSet$ and
for all $\phi\in\Env$ that
$\AppBound{\phi}{\cur{C}}=\AppBound{\phi}{\hat{\cur{C}}}$ and
$\AppUnbound{\phi}{\cur{C}}=\AppUnbound{\phi}{\hat{\cur{C}}}$.
Thus 
$\cur{C}$ and $\hat{\cur{C}}$ give the same division of $\partial\phi(\man{M})$,
for all $\phi\in\Env$ yet, $\cur{C}\cap\hat{\cur{C}}=\varnothing$. Thus they
are not related by the standard algebra of sets, due to the reparameterization
of the curves in \cur{C}.

Second, since the definition of the b.p.p.\ can rely on non-precompact curves
we have the following. Given $\cur{C}\in\BPPSet$ and $\gamma\in\cur{C}- f(\cur{C})$
let $\cur{D}=\left(\cur{C}-{\gamma}\right)\cup\{\gamma(2t)\}$. By construction
$f(\cur{D})=f(\cur{C})$ and therefore 
$\AppBound{\phi}{\cur{C}}=\AppBound{\phi}{\cur{D}}$ and $\AppUnbound{\phi}{\cur{C}}=
\AppUnbound{\phi}{\cur{D}}$. However, $\cur{D}\not\subset\cur{C}$ and $\cur{C}\not\subset\cur{D}$.
That is, despite $\cur{C}$ and $\cur{D}$ giving the same division of $\partial\phi(\man{M})$,
for all $\phi\in\Env$, they are not related by inclusion. In addition, if there
exists $p\in P({\gamma})-P(\cur{C}-\{\gamma\})$ then $\cur{C}\cap\cur{D}\not\in\BPPInexCom$, despite 
$\gamma$ not contributing to the classification of boundary points as it is precompact.

Third, again as the definition of the b.p.p.\ can rely on non-precompact curves,
we have the following. 
For this example we create an alternative to $g$ by changing our choice of the
curves $\gamma_p$.
We can choose different vectors, for example $\hat{v}_p\in T_p\man{M}$,
to produce a different collection of geodesics
$\hat{\gamma}_p:[-\hat{\epsilon}_p,
\hat{\epsilon}_p)\to \man{M}$ so that $\hat{\gamma}_p(0)=p$ and 
$\hat{\gamma}'_p(0)=\hat{v}_p$. By replacing the $\gamma_p$'s by the
$\hat{\gamma}_p$'s we can use Definitions 
\ref{Cnormdef} and \ref{curves less than} and Proposition \ref{gfunc}
to construct a function $\hat g:\BPPInexCom\to\BPPSet$. 
Lemma \ref{fdinverse} and Propositions \ref{prop.goodtranslation},
\ref{lem.fdinvers} and 
\ref{prop:rel3} remain true once $g$ is replaced by $\hat g$ and therefore
$\hat g$ is a viable alternative to $g$ despite $\hat g\neq g$.
None-the-less, for any $S\in\BPPInexCom$, 
it is the case that 
$\AppBound{\phi}{g({S})}=\AppBound{\phi}{\hat{g}({S})}$ and
$\AppUnbound{\phi}{g({S})}=\AppUnbound{\phi}{\hat{g}({S})}$.
Hence $g(S)$ and $\hat{g}(S)$ induce the same division of $\partial\phi(\man{M})$,
for all $\phi\in\Env$. If $\man{M}-P(S^\leq)\neq\EmptySet$ then
$g(S)\not\subset \hat{g}(S)$, $\hat{g}(S)\not\subset g(S)$
and $g(S)\cap \hat{g}(S)=S^\leq\not\in\BPPSet$,
due to the differences in the construction of $g$ and $\hat{g}$. These
differences
only concern precompact curves: curves which have no effect on the classification
of boundary points.

\section{Creating a one-to-one correspondence}\label{Sc one-to-one correspondences}

Propositions \ref{lem.fdinvers}, \ref{prop:rel3} and Section \ref{subsec.examples} suggest that
the sets $\BPPInexCom$ and $\BPPSet$ contain more information than 
needed for the abstract boundary classification. This suggests that we should
define an equivalence relation.

\subsection{A curve based equivalence relation}
\label{ssec:ambiguity}

Since it is the elements of $f(\cur{C})$, $\cur{C}\in\BPPSet$, which
influence the division of $\partial\phi(M)$, for any $\phi\in\Env$,
we start by defining an equivalence relation on $\BPPInexCom$. 

Given $S,P\in\BPPInexCom$, as only the relative bounded/ un-boundedness
of pairs of curves is important, we design our equivalence relation so that
$S$ is equivalent to $P$ if and only if the images of all curves 
in $S$ is equal to the images of all curves in $P$ and that the boundedness and unboundedness
of curves that eventually have the same image are the same.

\begin{Def}\label{def.equiv2}
	Define an equivalence relation $\simeq_{\textnormal c}$ on the set $\BPPInexCom$ by $S\simeq_{\textnormal c} P$ if and only if
	\[
		\forall\gamma:[a,b)\to\man{M}\in S,\ \exists\delta:[p,q)\to\man{M}\in P,\ c,r\in\mathbb{R}
	\] 
	so that $\gamma([c,b))=\delta([r,q))$ and either both $\gamma,\delta$ are bounded or unbounded,
	and,
	\[
		\forall\delta:[p,q)\to\man{M}\in P,\ \exists\gamma:[a,b)\to\man{M}\in S,\ r,c\in\mathbb{R}
	\] so that $\delta([r,q))=\gamma([c,b))$ and either both $\delta,\gamma$ are bounded or unbounded.
	We shall write the equivalence class of $S$ by $[S]_{\textnormal c}$.
\end{Def}
    
\begin{Lem}
	The equivalence relation $\simeq_{\textnormal c}$ on $\BPPInexCom$ is well defined.
\end{Lem}
\begin{proof}
	Let $S\in\BPPInexCom$ then as $\gamma\leq\gamma$ for all $\gamma\in S$ 
	we can see that $S\simeq_{\textnormal c} S$.
	It is clear that the symmetry of $\simeq_{\textnormal c}$ is satisfied by definition.
	
	Suppose that $S,P,Q\in\BPPInexCom$ are such that $S\simeq_{\textnormal c} P$ and 
	$P\simeq_{\textnormal c} Q$.  Let $\gamma:[a,b)\to\man{M}\in S$
	then there exists $\delta:[p,q)\to\man{M}\in P$ and $c,r\in\mathbb{R}$ 
	so that $\gamma([c,b))=\delta([r,q))$ and either both
	$\gamma$ and $\delta$ are bounded or unbounded.  Likewise as 
	$\delta\in P$ there exists $\mu:[u,v)\to\man{M}\in Q$ and $s,w\in\mathbb{R}$
	so that $\delta([s,q))=\mu([w,v))$ and either both are bounded or unbounded. 
	Without loss of generality
	assume that $r<s$, then as $\gamma([c,b))=\delta([r,q))$ there must exist 
	$d\in\mathbb{R}$ so that 
	$\gamma([d,b))=\delta([s,q))=\mu([w,v))$.  If $\gamma$ is bounded then 
	$\delta$ must be bounded and therefore $\mu$ must also be bounded.  
	Likewise, if $\gamma$ is unbounded then $\mu$ must be unbounded. 
	A similar argument can be applied to any $\gamma\in Q$,
	hence $S\simeq_{\textnormal c} Q$.
	
	Therefore $\simeq_{\textnormal c}$ is well defined.
\end{proof}
    
\begin{Lem}\label{lem. simeeq 2 equiv}
	Let $S\in\BPPInexCom$ then $S\simeq_{\textnormal c} f\circ g(S)$. 
\end{Lem}
\begin{proof}
	This follows from Lemma \ref{fdinverse} and Definition \ref{def.equiv2}.
\end{proof}
    
\begin{Def}\label{def.secondinducedequivalence}
	Define $\cur{C},\cur{D}\in \BPPSet$ to 
	be equivalent, denoted $\cur{C}\approx_{\textnormal c}\cur{D}$, if and only if $f(\cur{C})\simeq_{\textnormal c} f(\cur{D})$.  It is clear 
	that this provides a well-defined equivalence relation. Denote the equivalence class of $\cur{C}$ by $[\cur{C}]_{\textnormal c}$.
\end{Def}

The correspondence induced by $f$ and $g$ under $\simeq_{\textnormal c}$ and $\approx_{\textnormal c}$
is one-to-one.
    
\begin{Prop}\label{cequivprop}
	The induced functions $${f_{\textnormal c}}:\frac{BPP(\man{M})}{\approx_{\textnormal c}}\to\frac{\BPPInexCom}{\simeq_{\textnormal c}}$$
	and $${g_{\textnormal c}}:\frac{\BPPInexCom}{\simeq_{\textnormal c}}\to\frac{BPP(\man{M})}{\approx_{\textnormal c}}$$
	are bijective and mutually inverse.
\end{Prop}
\begin{proof}
	We first show that ${f_{\textnormal c}}$ and ${g_{\textnormal c}}$ are well defined.			
	Let $\cur{C}\approx_{\textnormal c}\cur{D}$, we must show that $f(\cur{C})\simeq_{\textnormal c} f(\cur{D})$.  			
	This is true by definition and
	therefore ${f_{\textnormal c}}$ is well-defined. Now let $S\simeq_{\textnormal c} P$,
	from Lemma \ref{lem. simeeq 2 equiv} we know that $f\circ g(S)\simeq_{\textnormal c} f\circ g(P)$. By Definition
	\ref{def.secondinducedequivalence} this implies that $g(S)\approx_{\textnormal c} g(P)$. Thus $g_{\textnormal c}$ is well defined.
	
	We now show that ${f_{\textnormal c}}\circ{g_{\textnormal c}}([S]_{\textnormal c})=[S]_{\textnormal c}$ and ${g_{\textnormal c}}\circ{f_{\textnormal c}}([\cur{C}]_{\textnormal c})=[\cur{C}]_{\textnormal c}$. 
	Let $[\cur{C}]_{\textnormal c}\in\frac{BPP(\man{M})}{\approx_{\textnormal c}}$ and note that ${f}\circ{g}\circ{f}(\cur{C})={f}(\cur{C})$
	by Proposition \ref{lem.fdinvers}.
	By Definition \ref{def.secondinducedequivalence} we know that $g\circ f(\cur{C})\approx_{\textnormal c}\cur{C}$ or rather that 
	${g_{\textnormal c}}\circ{f_{\textnormal c}}([\cur{C}]_{\textnormal c})=[\cur{C}]_{\textnormal c}$.
	Let $S\in\BPPInexCom$ then by Lemma \ref{lem. simeeq 2 equiv} we know $S\simeq_{\textnormal c} f\circ g(S)$. That is
	$f_{\textnormal c}\circ g_{\textnormal c}([S]_{\textnormal c})=[S]_{\textnormal c}$ as required. This is sufficient
	to prove that $f_{\textnormal c}$ and $g_{\textnormal c}$ are bijective and mutually inverse as required.
\end{proof}

Thus by removing the `additional information' introduced by $\leq$ and illuminated by
the 
construction of $g$ we can put $\BPPInexCom$ and $\BPPSet$ into a one-to-one 
correspondence. That is, every element of $\BPPSet$ is in the equivalence
class, under $\approx_{\textnormal c}$, of the image, under $g$, of some element of $\BPPInexCom$.

Note, however, that these equivalence relations do not ensure that 
$\cur{C}\not\approx_{\textnormal c}\cur{D}$ implies 
$$\AppBound{\phi}{\cur{C}}\neq\AppBound{\phi}{{\cur{D}}}$$ and
$$\AppUnbound{\phi}{\cur{C}}\neq\AppUnbound{\phi}{{\cur{D}}}.$$ An explicit
example can be constructed using the two classes of null 
geodesics and one of the
two maximal embeddings 
of the Misner spacetime (for $0<t<\infty$) discussed in Section 5.8 of
\cite{HawkingEllis1973}.

 \begin{Ex}\label{example}
   Let $\man{M}=\mathbb{R}^+\times S^1$ be the Misner spacetime, 
   \cite{HawkingEllis1973}, given in the coordinates $0<t<\infty$ and 
   $0\leq\psi< 2\pi$. 
   With
   respect to these coordinates we have $ds^2=-t^{-1}dt^2+td\psi^2$. 
   Let $\phi:\man{M}\to\mathbb{R}\times S^1$ be defined by
   $\phi(t,\psi)=(t,\psi-\log(t))$. This is one of the two maximal 
   extensions of $\man{M}$. 
   Letting $\psi' = \psi-\log(t)$ the metric becomes 
   $ds^2=2d\psi'dt+t\left(d\psi'\right)^2$.
   With respect to this extension we are able to select two sets of 
   affinely parametrised
   null geodesics lying in $\man{M}$: those that
   run vertically and those that 
   approach but never reach the waist $t=0$. 
   
   Let $\cur{C}_{v}$ 
   be the set of all affinely parametrised null geodesics so that $\phi(\cur{C}_v)$ is the
   set of vertical null geodesics in the extension given by $\phi$. Likewise let $\cur{C}_w$ 
   be the set of all affinely parametrised null geodesics so that $\phi(\cur{C}_w)$ is the    
   set of null geodesics that approach but do not reach the waist. 
   Thus $\gamma\in \cur{C}_v$ 
   is given by $\phi\circ\gamma(\tau)=(a_0\tau+b_0,b_1)$ for $a_0,b_0\in\mathbb{R}$
   and $0\leq b_1<2\pi$ where $\tau\in(-\frac{b_0}{a_0},\infty)$. Similarly $\gamma\in \cur{C}_w$ is given by 
   $\phi\circ\gamma(\tau)=\left(a_0\tau+b_0,-2\log(a_0\tau+b_0)+b_1)\right)$ for $a_0\in\mathbb{R}^+,b_0\in\mathbb{R}$ and $0\leq b_1<2\pi$ where $\tau\in\left(-\frac{b_0}{a_0},\infty\right)$, \cite{Whale2010}. These formulas and
   Example 5 of \cite{ScottSzekeres1994} imply that $\cur{C}_v,\cur{C}_W\in\BPPSet$.

   Let $\gamma\in\cur{C}_v$ and $\mu\in\cur{C}_w$. The equation $\phi\circ\gamma(\tau)=\phi\circ\mu(\tau')$ has the solution
   \begin{align*}
     \tau_k &= \frac{1}{a_0}\left(\exp\left(\frac{b_1-b_1'+k\pi}{-2}\right)-b_0\right)\\
     \tau_k' &=\frac{1}{a_0'}\left(a_0\tau_k+b_0-b_0'\right),
   \end{align*}
   for each $k\in\mathbb{Z}$. As $k\to\infty$ we see that $\tau_k\to\frac{-b_0}{a_0}$ and $\tau_k'\to\frac{-b_0'}{a_0'}$.
   The sequence $\{(t_k,\psi'_k)=\phi\circ\gamma(\tau_k)=\phi\circ\mu(\tau_k')\}$ has, 
   as $k\to\infty$, the end point $(0,b_1)$, as expected.
   As $\gamma$ and $\mu$ are arbitrary, and $\phi$ is a diffeomorphism, this calculation implies 
   that given $\psi\in\Env$, a curve $\gamma\in\cur{C}_v$ has $p\in\partial\psi(\man{M})$ as a limit point if and only if
   every curve in $\cur{C}_w$ has $p$ as a limit point and that $\mu\in\cur{C}_w$ has $p\in\partial\psi(\man{M})$ as an endpoint if and only if
   every curve in $\cur{C}_v$ has $p$ as a limit point.
   
   Let $\psi\in\Env$.
   Since every curve in $\cur{C}_v$ and $\cur{C}_w$ has bounded 
   parameter we know that $\AppUnbound{\psi}{\cur{C}_v}=\AppUnbound{\psi}{\cur{C}_w}=\varnothing$.
   From the 
   calculation above we know that $p\in\AppBound{\psi}{\cur{C}_v}$ if and only if $p\in\AppBound{\psi}{\cur{C}_w}$. 
   That is, $\AppBound{\psi}{\cur{C}_v}=\AppBound{\psi}{\cur{C}_w}$.
   Hence both $\cur{C}_v$ and $\cur{C}_w$ induce the same division of $\partial\psi(\man{M})$ for
   any $\psi\in\Env$. By construction, however, $\cur{C}_c\not\approx_c\cur{C}_w$.
 \end{Ex}
 
	This implies that we should look for a second pair of equivalence relations.

\subsection{A boundary based equivalence relation}
\label{The Obvious Equivalence Relation}

Equating elements of $\BPPSet$ that produce the same 
division of $\partial\phi(M)$, for any $\phi\in\Env$,
will ensure that each equivalence class gives a 
unique classification of abstract boundary points.
      
\begin{Def}\label{approx 1}
  Let $\approx_{\textnormal b}$ be the equivalence relation on $\BPPSet$ given by 
  $\cur{C}\approx_{\textnormal b}\cur{D}$ if and only if
  for all $\phi\in\Env$, $\AppUnbound{\phi}{\cur{C}}=\AppUnbound{\phi}{\cur{D}}$ and 
  $\AppBound{\phi}{\cur{C}}=\AppBound{\phi}{\cur{D}}$.
  We denote the equivalence class of $\cur{C}$ by $[\cur{C}]_{\textnormal b}$.
\end{Def}

We will also need to ensure that $f(\cur{C})$ is equated with
$f(\cur{D})$ when $\cur{C}\approx_{\textnormal b}\cur{D}$.

\begin{Def}\label{def.first Equivalence Relation Induced}
  Let $\simeq_{\textnormal b}$ be the equivalence relation on $\BPPInexCom$ given by 
  $S\simeq_{\textnormal b} P$ if and only if $g(S)\approx_{\textnormal b} g(P)$.
  We denote the equivalence class of $S$ by $[S]_{\textnormal b}$.
\end{Def}

Both equivalence relations are clearly well defined.

We need the following lemma before proving that the correspondence induced by 
$f$ and $g$ under $\simeq_{\textnormal b}$ and $\approx_{\textnormal b}$
is one-to-one.

\begin{Lem}\label{lem.gcircfEquivalence}
	Let $\cur{C}\in\BPPSet$ then $\cur{C}\approx_{\textnormal b}g\circ f(\cur{C})$.
\end{Lem}
\begin{proof}
  Let $\cur{D}=g\circ f(\cur{C})$.
	From Proposition \ref{lem.fdinvers} we know that $f(\cur{D})=f(\cur{C})$. Thus 
	$\cur{D}\cap\InexCom=\cur{C}\cap\InexCom$. 	
	Let $p\in \App{\phi}{\cur{C}}$
	then there exists $\gamma\in \cur{C}$
	so that $p$ is a limit point of $\phi\circ\gamma$.
	Since 
	$p\in\partial\phi(\man{M})$ the curve $\gamma$ is non-precompact.
	Therefore $\gamma\in\cur{C}\cap\InexCom$. This implies that
	$\gamma\in\cur{D}$ and hence that $p\in\App{\phi}{\cur{D}}$.
	This argument can be applied to $p\in\App{\phi}{\cur{D}}$ to show that
	\[
	  \left\{\gamma\in\cur{C}:p\text{ is a limit point of }\phi\circ\gamma\right\}
	  =
	  \left\{\gamma\in\cur{D}:p\text{ is a limit point of }\phi\circ\gamma\right\}
	  .
	\]
	This implies, for all $\phi\in\Phi$, that 
	$\AppUnbound{\phi}{\cur{C}}=\AppUnbound{\phi}{\cur{D}}$ and
	$\AppBound{\phi}{\cur{C}}=\AppBound{\phi}{\cur{D}}$.
	Thus $\cur{C}\approx_{\textnormal b}f\circ g(\cur{C})$.
\end{proof}

\begin{Thm}
	The induced functions $f_{\textnormal b}:\frac{\BPPSet}{\approx_{\textnormal b}}\to\frac{\InexCom}{\simeq_{\textnormal b}}$ and 
	$g_{\textnormal b}:\frac{\InexCom}{\simeq_{\textnormal b}}\to\frac{\BPPSet}{\approx_{\textnormal b}}$ are bijective and mutually inverse.
\end{Thm}
\begin{proof}
	We need to show that $f_{\textnormal b}$ and $g_{\textnormal b}$ are well-defined.			
	Suppose that $\cur{C}\approx_{\textnormal b}\cur{D}$ we need to show that 
	$f(\cur{C})\simeq_{\textnormal b} f(\cur{D})$. That is we need to show that
	$g\circ f(\cur{C})\approx_{\textnormal b} g\circ f(\cur{D})$. From Lemma 
	\ref{lem.gcircfEquivalence}, however, we know that
	$g\circ f(\cur{C})\approx_{\textnormal b}\cur{C}\approx_{\textnormal b}\cur{D}\approx_{\textnormal b}
	 g\circ f(\cur{D})$ as required. 
	Likewise, suppose that $S\simeq_{\textnormal b} P$ we need to show that 
	$g(S)\approx_{\textnormal b} g(P)$, but this follows directly from Definition 
	\ref{def.first Equivalence Relation Induced}. Hence $f_{\textnormal b}$ and 
	$g_{\textnormal b}$ are well defined.
	
	We now show that $f_{\textnormal b}\circ g_{\textnormal b}([S]_{\textnormal b})=[S]_{\textnormal b}$ and 
	$g_{\textnormal b}\circ f_{\textnormal b}([\cur{C}]_{\textnormal b})=[\cur{C}]_{\textnormal b}$.
	Let $[\cur{C}]_{\textnormal b}\in\frac{\BPPSet}{\approx_{\textnormal b}}$, from 
	Lemma \ref{lem.gcircfEquivalence} we know that
	$g\circ f(\cur{C})\approx_{\textnormal b} \cur{C}$ so that
	$
					g_{\textnormal b}\circ f_{\textnormal b}([\cur{C}]_{\textnormal b})
									=[\cur{C}]_{\textnormal b},
	$
	as required.
	Let $S\in\BPPInexCom$. From Lemma \ref{lem.gcircfEquivalence} we see that
	$g\circ f\circ g(S)\approx_{\textnormal b} g(S)$. 
	By Definition \ref{def.first Equivalence Relation Induced}
	this implies that $f\circ g(S)\simeq_{\textnormal b} S$ or that 
	$f_{\textnormal b}\circ g_{\textnormal b}([S])=[S]_{\textnormal b}$ as required.
	Therefore $f_{\textnormal b}$ and $g_{\textnormal b}$ are both bijective are mutually inverse.
\end{proof}

While $\approx_{\textnormal b}$ appropriately encodes when two
b.p.p.\ satisfying sets of curves produce the same 
division of $\partial\phi(M)$, for any $\phi\in\Env$, it is substantially harder to work with than $\approx_{\textnormal c}$.
Given $\cur{C},\cur{D}\in\BPPSet$, in order to determine if
$\cur{C}\approx_{\textnormal b}\cur{D}$ we need to check if 
$\AppUnbound{\phi}{\cur{C}}=\AppUnbound{\phi}{\cur{D}}$ and 
  $\AppBound{\phi}{\cur{C}}=\AppBound{\phi}{\cur{D}}$ for all
$\phi\in\Phi(\man{M})$. This is, in general, extremely hard.
The relation $\approx_{\textnormal c}$ is much easier to do calculations with. It is
fortunate then that $\approx_{\textnormal c}$ and $\approx_{\textnormal b}$ are closely related.

 \begin{Prop}\label{relation between equivalence relations}
   Let $\cur{C},\cur{D}\in \BPPSet$ then $\cur{C}\approx_{\textnormal c}\cur{D}$ 
   implies that 
   $\cur{C}\approx_{\textnormal b}\cur{D}$.
 \end{Prop}
 \begin{proof}
   Suppose that $\cur{C}\approx_{\textnormal c}\cur{D}$ and let $p\in\App{\phi}{\cur{C}}$. 
   Then there exists
   $\gamma\in\cur{C}$ so that $p$ is an accumulation point of 
   $\phi\circ \gamma$.  Since $p\in\partial\phi(\man{M})$ it is  the
   case that $\gamma\in f(\cur{C})$. Since $\cur{C}\approx_{\textnormal c}\cur{D}$ we 
   know that $f(\cur{C})\simeq_{\textnormal c} f(\cur{D})$. Hence there exists
   $\delta\in f(\cur{D})$ so that the images of $\gamma$ and $\delta$ agree, 
   except on some compact portion.  Thus
   $p$ is an accumulation point of $\phi\circ\delta$ and therefore 
   $p\in\App{\phi}{\cur{D}}$. Moreover, if $\gamma$ is bounded (unbounded)
   then, by Definition
   \ref{def.equiv2}, $\delta$ must also be bounded (unbounded). Therefore 
   $\AppBound{\phi}{\cur{C}}\subset\AppBound{\phi}{\cur{D}}$
   and $\AppUnbound{\phi}{\cur{C}}\subset\AppUnbound{\phi}{\cur{D}}$.
   The same argument can be applied
   to $p\in\App{\phi}{\cur{D}}$ and therefore 
   $\cur{C}\approx_{\textnormal b}\cur{D}$.
 \end{proof}

Thus we have that $\approx_{\textnormal c}\subset\approx_{\textnormal b}$.
The example mentioned after Proposition \ref{cequivprop} demonstrates that
$\approx_{\textnormal b}\not\subset\approx_{\textnormal c}$.

\section{A generalized algebra of sets on the set of all b.p.p. satisfying sets of curves}\label{sec:realtions}

To remedy the problems described in Section \ref{subsec.examples} we
now give a generalization of the usual algebra of sets on $\BPPSet$ using the results
of Sections \ref{sec.main} and \ref{Sc one-to-one correspondences}. Our aim
is to produce a new algebra of sets that has a stronger connection with
the division of $\partial\phi(M)$ induced by elements of $\BPPSet$.

We use $\simeq_{\textnormal b}$ and $\approx_{\textnormal b}$ below. The same results hold, with the
exclusion of the ``if'' part of Proposition \ref{prop.subsetbppplaysnice}
when replacing $\simeq_{\textnormal b}$ and $\approx_{\textnormal b}$ with $\simeq_{\textnormal c}$ and $\approx_{\textnormal c}$,
respectively.
  
  \begin{Def}\label{Modified set realtion definitions.subset}
    Let $\cur{C},\cur{D}\in BPP(\man{M})$.
    We say that $\cur{C}$ is a subset of $\cur{D}$, denoted 
    $\cur{C}\subset_{\bpp.}\cur{D}$, if and only if
    there exists $S\in[{f}(\cur{D})]_{\textnormal b}$ so that
    ${f}(\cur{C})\subset S$.
  \end{Def}
  
  This ensures the following;
  
  \begin{Prop}\label{prop.subsetbppplaysnice}
    Let $\cur{C},\cur{D}\in\BPPSet$ then
    $\cur{C}\subset_{b.p.p.}\cur{D}$ if and only if
    for all $\phi\in\Env$, $\App{\phi}{\cur{C}}\subset\App{\phi}{\cur{D}}$ and
    $\AppBound{\phi}{\cur{C}}\subset\AppBound{\phi}{\cur{D}}$.
  \end{Prop}
  \begin{proof}
		Assume that $\cur{C}\subset_{b.p.p.}\cur{D}$.
    Let $p\in\App{\phi}{\cur{C}}$ then there exists $\gamma\in\cur{C}$
    so that $p$ is approached by $\phi\circ\gamma$, hence $\gamma\in f(\cur{C})$.
    Since
    $\cur{C}\subset_{b.p.p.}\cur{D}$ there exists
    $S\in[f(\cur{D})]_{\textnormal b}$ so that $f(\cur{C})\subset S$.
    Thus $\gamma\in S$, so that $p\in\App{\phi}{g(S)}$.
    By Definition \ref{def.first Equivalence Relation Induced}
    and Lemma \ref{lem.gcircfEquivalence}
    we see that $g(S)\approx_{\textnormal b} g\circ f(\cur{D})\approx_{\textnormal b} D$.
		From Definition \ref{approx 1} we see that
		$\App{\phi}{\cur{C}}\subset\App{\phi}{\cur{D}}$. If
		$p\in\AppBound{\phi}{\cur{C}}$ then $\gamma$ can be taken to be bounded
		so that $p\in\AppBound{\phi}{\cur{D}}$.

		Assume that for all $\phi\in\Env$, $\App{\phi}{\cur{C}}\subset\App{\phi}{\cur{D}}$ and
    $\AppBound{\phi}{\cur{C}}\subset\AppBound{\phi}{\cur{D}}$. For each
    $\phi\in\Env$ and $p\in\App{\phi}{\cur{D}}-\App{\phi}{\cur{C}}$
    we can choose $\mu_{\phi,p}\in f(\cur{D})$ so that $\phi\circ\mu_{\phi,p}$
    approaches $p$ and $\mu_{\phi,p}$ is bounded if $p\in\AppBound{\phi}{\cur{D}}$. Let 
    \[
      S=f(\cur{C})\cup\left\{\mu_{\phi,p}\in f(\cur{D}):\phi\in\Env,\, p\in\App{\phi}{\cur{D}}-\App{\phi}{\cur{C}}\right\}.
    \]

    We now show that $S\in\BPPInexCom$. Let $\gamma:[a,b)\to \man{M},\,\delta:[p,q)\to\man{M}\in S$ so that
    there exists $c\in[a,b)$ and $r\in[p,q)$ so that $\gamma([c,b))=\delta([r,q))$. This implies that
    for all $\phi\in\Env$ the limit points of $\phi\circ\gamma$ are the same as the limit points of
    $\phi\circ\delta$. In turn
    this, together with the definition of $S$ implies that either, 
    $\delta,\gamma\in f(\cur{C})$ or 
    $\delta,\gamma\in f(\cur{D})$.
    In either case $\delta$ and $\gamma$ are either both bounded or both 
    unbounded as
    $f(\cur{C}),f(\cur{D})\in\BPPInexCom$. This implies that
    $S\in\BPPInexCom$.
    
    By construction we have that
    $\App{\phi}{g(S)}=\App{\phi}{\cur{D}}$, $\AppBound{\phi}{g(S)}=\AppBound{\phi}{\cur{D}}$, for all $\phi\in\Env$
    and that $f(\cur{C})\subset S$. That is $\cur{C}\subset_{b.p.p.}\cur{D}$ as required.
  \end{proof}

  \begin{Cor}
    Let $\cur{C},\cur{D}\in\BPPSet$ then $\cur{C}\subset_{b.p.p.}\cur{D}$
    and $\cur{D}\subset_{b.p.p.}\cur{C}$ implies that $\cur{C}\approx_{\textnormal b}
    \cur{D}$.
  \end{Cor}
  \begin{proof}
    From Proposition \ref{prop.subsetbppplaysnice} we know that
    $\App{\phi}{\cur{C}}=\App{\phi}{\cur{D}}$
    and $$\AppBound{\phi}{\cur{C}} = \AppBound{\phi}{\cur{D}}$$
    which implies the result.
  \end{proof}

  Given $\cur{C},\cur{D}\in\BPPSet$ it can be the case that
  $f(\cur{C})\cup f(\cur{D})\not\in\BPPSet$.
  For example choose 
	$\gamma_i:\left[0,\infty\right)\to\man{M}\in\InexCom$, $i=0,1$,
	two 
	curves in $M$, that eventually have different images.
	Define $s:\left[0,\infty\right)\to\left[0,1\right)$
	by $s(x)=\frac{2}{\pi}\tan^{-1}\left(x\right)$.
	Let $\cur{C},\cur{D}\in\BPPSet$ be defined by
	$\cur{C}=g(\{\gamma_0,\gamma_1\})$ and 
	$\cur{D}=g(\{\gamma_0\circ s,\gamma_1\})$.
	The union of $\cur{C}$ and $\cur{D}$ makes little sense as
	$\gamma_0$ and $\gamma_0\circ s$ have incompatible parameters, i.e.
	the condition of Definition \ref{def.overline} does not hold on
	$f(\cur{C})\cup f(\cur{D})$. 
	Thus the union can only be defined when 
	$f(\cur{C})\cup f(\cur{D})\in\BPPSet$.
	
	\begin{Def}\label{Modified set realtion definitions.union}
    Let $\cur{C},\cur{D}\in BPP(\man{M})$.
	  If $f(\cur{C})\cup f(\cur{D})\in\BPPInexCom$ then
	  the union of
	  $\cur{C}$ and $\cur{D}$ is given
	  by $g(f(\cur{C})\cup f(\cur{C}))$. 
	  Otherwise the union cannot be defined. The union
	  of $\cur{C}$ and $\cur{D}$ is denoted $\cur{C}\cup_{b.p.p.}\cur{D}$.
  \end{Def}
	
	It is not necessarily the case that $\cur{C},\cur{D}\subset 
	\cur{C}\cup_{b.p.p.}\cur{D}$. We do have, however, the following;
	
	\begin{Prop}\label{prop.subsetandcapbppplaynice}
	  Let $\cur{C},\cur{D}\in\BPPSet$ so that $\cur{C}\cup_{b.p.p.}\cur{D}$
	  can be defined then 
	  $\cur{C},\cur{D}\subset_{b.p.p.}\cur{C}\cup_{b.p.p.}\cur{D}$.
	\end{Prop}
	\begin{proof}
		By Proposition \ref{lem.fdinvers} we have that
    \begin{align*}
      f(\cur{C}\cup_{b.p.p.}\cur{D})&=f\circ g(f(\cur{C})\cup f(\cur{D}))\\
       &= f\circ g\circ f(\cur{C})\cup f\circ g\circ f(\cur{C})\\
       &= f(\cur{C})\cup f(\cur{D}),
    \end{align*}
		so that $f(\cur{C})\cup f(\cur{D})\in[f(\cur{C}\cup_{b.p.p.}\cur{D})]_{\textnormal b}$.
		Since $f(\cur{C}),f(\cur{D})\subset f(\cur{C})\cup f(\cur{D})$ 
		we have our result.
	\end{proof}
	
	\begin{Cor}
	  Let $\cur{C},\cur{D}\in\BPPSet$ so that $\cur{C}\cup_{b.p.p.}\cur{D}$
	  can be defined then 
	  \begin{gather*}
	    \App{\phi}{\cur{C}},\, \App{\phi}{\cur{D}}\subset
	      \App{\phi}{\cur{C}\cup_{b.p.p.}\cur{D}}\\
			\AppBound{\phi}{\cur{C}},\, \AppBound{\phi}{\cur{D}}\subset
	      \AppBound{\phi}{\cur{C}\cup_{b.p.p.}\cur{D}}
	  \end{gather*}
	\end{Cor}
	\begin{proof}
	  This follows from Propositions \ref{prop.subsetbppplaysnice}
	  and \ref{prop.subsetandcapbppplaynice}
	\end{proof}

	Given $\cur{C},\cur{D}\in\BPPSet$ the set $f(\cur{C})\cap f(\cur{D})$
	is always an element of $\BPPInexCom$.
	
	\begin{Def}\label{Modified set realtion definitions.intersection}
    Let $\cur{C},\cur{D}\in BPP(\man{M})$.
    The intersection of $\cur{C}$ and $\cur{D}$ is
    defined by $g\left(f(\cur{C})\cap f(\cur{D})\right)$.
    We denote the intersection by $\cur{C}\cap_{\bpp.}\cur{D}$.
  \end{Def}
  
  Just as for $\cup_{b.p.p.}$ it maybe the case that
  $\cur{C}\cap_{b.p.p.}\cur{D}\not\subset\cur{C},\cur{D}$. We do have 
  the following, however;
  
  \begin{Prop}\label{prop.capb.p.p.subsetplaynice}
    Let $\cur{C},\cur{D}\in\BPPSet$, then 
	  $\cur{C}\cap_{b.p.p.}\cur{D}\subset_{b.p.p.}\cur{C},\cur{D}$.
  \end{Prop}
  \begin{proof}
    By Proposition \ref{lem.fdinvers} we have that
    \begin{align*}
      f(\cur{C}\cap_{b.p.p.}\cur{D})&=f\circ g(f(\cur{C})\cap f(\cur{D}))\\
       &= f\circ g\circ f(\cur{C})\cap f\circ g\circ f(\cur{C})\\
       &= f(\cur{C})\cap f(\cur{D}).
    \end{align*}
		Since $f(\cur{C})\cap f(\cur{D})\subset f(\cur{C}),f(\cur{D})$
		we have our result.
  \end{proof}

  \begin{Cor}
    Let $\cur{C},\cur{D}\in\BPPSet$  then 
	  \begin{gather*}
	    \App{\phi}{\cur{C}\cap_{b.p.p.}\cur{D}}\subset
	      \App{\phi}{\cur{C}},\, \App{\phi}{\cur{D}}\\
	    \AppBound{\phi}{\cur{C}\cap_{b.p.p.}\cur{D}}\subset
	      \AppBound{\phi}{\cur{C}},\, \AppBound{\phi}{\cur{D}}
	  \end{gather*}
  \end{Cor}
  \begin{proof}
    This follows from Propositions \ref{prop.subsetbppplaysnice}
    and \ref{prop.capb.p.p.subsetplaynice}.
  \end{proof}

	The operations $\cup_{b.p.p.}$ and $\cap_{b.p.p.}$ have the usual 
	relations with respect to each other.	
	\begin{Prop}
	  Let $\cur{C},\cur{D},\cur{E}\in\BPPSet$. Then
	  \begin{gather*}
	    \cur{C}\cup_{b.p.p.}\cur{D}\approx_{\textnormal b}\cur{D}\cup_{b.p.p.}\cur{C}\\
	    \cur{C}\cap_{b.p.p.}\cur{D}\approx_{\textnormal b}\cur{D}\cap_{b.p.p.}\cur{C}\\
	    (\cur{C}\cup_{b.p.p.}\cur{D})\cup_{b.p.p.}\cur{E}\approx_{\textnormal b}
	      \cur{C}\cup_{b.p.p.}(\cur{D}\cup_{b.p.p.}\cur{E})\\
	    (\cur{C}\cap_{b.p.p.}\cur{D})\cap_{b.p.p.}\cur{E}\approx_{\textnormal b}
	      \cur{C}\cap_{b.p.p.}(\cur{D}\cap_{b.p.p.}\cur{E})\\
	    \cur{C}\cup_{b.p.p.}(\cur{D}\cap_{b.p.p.}\cur{E})\approx_{\textnormal b}
	      (\cur{C}\cup_{b.p.p.}\cur{D})\cap_{b.p.p.}
	      (\cur{C}\cup_{b.p.p.}\cur{E})\\
	    \cur{C}\cap_{b.p.p.}(\cur{D}\cup_{b.p.p.}\cur{E})\approx_{\textnormal b}
	      (\cur{C}\cap_{b.p.p.}\cur{D})\cup_{b.p.p.}
	      (\cur{C}\cap_{b.p.p.}\cur{E})
	  \end{gather*}
	\end{Prop}
	\begin{proof}
	  We give the proof of the last statement. All other claims can be
	  proved using a similar method.
	  Proposition \ref{lem.fdinvers} implies that
	  \begin{align*}
	    f(\cur{C}\cap_{b.p.p.}(\cur{D}\cup_{b.p.p.}\cur{E}))
	      &=f\circ g\Bigl(f(\cur{C})\cap f\circ g\Bigl(f(\cur{D})\cup f(\cur{E})\Bigr)\Bigr)\\
				& = f\circ g\circ f(\cur{C})\cap 
				  \Bigl(f\circ g\circ f(\cur{D})\cup f\circ g\circ f(\cur{E})\Bigr)\\
				& = \Bigl(f\circ g\circ f(\cur{C})\cap f\circ g\circ f(\cur{D})\Bigr)
				\cup\Bigl(f\circ g\circ f(\cur{C})\cap f\circ g\circ f(\cur{E})\Bigr)\\
				&= f\circ g(f(\cur{C})\cap f(\cur{D}))\cup 
				  f\circ g(f(\cur{C})\cap f(\cur{E}))\\
				&= f(\cur{C}\cap_{b.p.p.}\cur{D})\cup f(\cur{C}\cap_{b.p.p.} \cur{E})\\
				&= f\circ g\circ f(\cur{C}\cap_{b.p.p.}\cur{D})\cup 
				  f\circ g\circ f(\cur{C}\cap_{b.p.p.} \cur{E})\\
				&= f\circ g\Bigl( f(\cur{C}\cap_{b.p.p.}\cur{D})
				  \cup f(\cur{C}\cap_{b.p.p.} \cur{E})\Bigr)\\
				&= f\Bigl((\cur{C}\cap_{b.p.p.}\cur{D})\cup_{b.p.p.}
	      (\cur{C}\cap_{b.p.p.}\cur{E})\Bigr).
	  \end{align*}
	  Thus by Definition \ref{def.secondinducedequivalence}
	  we know that $\cur{C}\cap_{b.p.p.}(\cur{D}\cup_{b.p.p.}\cur{E})
	  \approx_{\textnormal c} (\cur{C}\cap_{b.p.p.}\cur{D})\cup_{b.p.p.}
	      (\cur{C}\cap_{b.p.p.}\cur{E})$.
	  Hence by Proposition \ref{relation between equivalence relations},
	  $
	    \cur{C}\cap_{b.p.p.}(\cur{D}\cup_{b.p.p.}\cur{E})\approx_{\textnormal b}
	      (\cur{C}\cap_{b.p.p.}\cur{D})\cup_{b.p.p.}
	      (\cur{C}\cap_{b.p.p.}\cur{E}),
	  $
	  as required.
	\end{proof}

  Thus we now have an algebra of sets on $\BPPSet$ that implies
  appropriate relations for the 
  division of $\partial\phi(M)$, for all $\phi\in\Env$,
  induced by elements of $\BPPSet$.
  This is a necessary first step before
  analysing how the classification changes when the set of b.p.p.
  satisfying curves changes. This analysis makes up the
  second part of this series of papers.

\section*{Acknowledgements}
The author was partially funded by Marsden grant UOO-09-022.

\bibliographystyle{unsrt}
\bibliography{AB_classification_I_Ben_Whale}

\begin{thebibliography}{10}

\bibitem{Geroch1968a}
R.~Geroch.
\newblock Local characterization of singularities in general relativity.
\newblock {\em J. Math. Phys.}, 9:450--465, 1968.

\bibitem{Schmidt1971}
B.~G. Schmidt.
\newblock A new definition of singular points in general relativity.
\newblock {\em Gen. Rel. Grav.}, 1(3):269--280, 1971.

\bibitem{GerochPenroseKronheimer1972}
R.~Geroch, R.~Penrose, and E.~H. Kronheimer.
\newblock Ideal points in space-time.
\newblock {\em Proc. Roy. Soc. Lond. Ser. A}, 327(1571):545--567, 1972.

\bibitem{Marolf2003New}
D.~Marolf and S.~F. Ross.
\newblock {A new recipe for causal completions}.
\newblock {\em Class. Quantum Grav.}, 20(18):4085--4117, September 2003.

\bibitem{citeulike:8211160}
J.~Flores.
\newblock The causal boundary of spacetimes revisited.
\newblock {\em Comm. Math. Phys.}, 276(3):611--643, 2007.

\bibitem{Flores2010Final}
J.~Flores, J.~Herrera, and M.~S{\'a}nchez.
\newblock On the final definition of the causal boundary and its relation with
  the conformal boundary.
\newblock {\em Adv. Theor. Math. Phys}, 15(4), 2011.
\newblock (Preprint arXiv:1001.3270v2).

\bibitem{ScottSzekeres1994}
S.~M. Scott and P.~Szekeres.
\newblock The abstract boundary---a new approach to singularities of manifolds.
\newblock {\em J. Geom. Phys.}, 13(3):223--253, 1994.

\bibitem{Geroch1868}
R.~Geroch.
\newblock {What is a singularity in general relativity?}
\newblock {\em Ann. Phys.}, 48:526--540, 1968.

\bibitem{Ashley2002a}
M.~J. S.~L. Ashley.
\newblock {\em Singularity Theorems and the Abstract Boundary Construction}.
\newblock PhD thesis, Department of Physics, Australian National University,
  2002.
\newblock {http://hdl.handle.net/1885/46055}.

\bibitem{Whale2010}
B.~E. Whale.
\newblock {\em Foundations of and Applications for the Abstract Boundary
  Construction in Space-time}.
\newblock PhD thesis, Department of Quantum Science, Australian National
  University, 2010.
\newblock {http://hdl.handle.net/1885/49393}.

\bibitem{HawkingEllis1973}
S.~W. Hawking and G.~F.~R. Ellis.
\newblock {\em The Large Scale Structure of Space-Time}.
\newblock Cambridge University Press, 1973.

\bibitem{Whale2012b}
B.~E. Whale.
\newblock The dependence of the abstract boundary classification on a set of
  curves {II}: How the classification changes when the boundary parameter
  property satisfying set of curves changes.
\newblock Submitted with this paper.

\bibitem{FamaClarke1998}
C.~J. Fama and C.~J.~S. Clarke.
\newblock A rigidity result on the ideal boundary structure of smooth
  spacetimes.
\newblock {\em Class. Quantum Grav.}, 15(9):2829--2840, 1998.

\bibitem{FamaScott1994}
C.~J. Fama and S.~M. Scott.
\newblock Invariance properties of boundary sets of open embeddings of
  manifolds and their application to the abstract boundary.
\newblock In {\em Differential Geometry and Mathematical Physics}, volume 170
  of {\em Contem. Math.}, pages 79--111. AMS, 1994.

\end{thebibliography}

\end{document}